\documentclass[11pt]{article}

\usepackage{xspace}
\usepackage{graphicx}
\usepackage{enumitem}
\usepackage{amsmath,amssymb, amsthm}
\usepackage{fullpage}

\newtheorem{theorem}{Theorem}[section]
\newtheorem{lemma}[theorem]{Lemma}

\newtheorem{Definition}[theorem]{Definition}
\newtheorem{claim}[theorem]{Claim}

\newcommand{\set}[1]{\left\{#1\right\}}
\newcommand{\cardinal}[1]{\left|#1\right|}
\newcommand{\floor}[1]{\left\lfloor#1\right\rfloor}
\newcommand{\ceil}[1]{\left\lceil#1\right\rceil}
\newcommand{\lfrac}[1]{\left\lfloor#1\right\rceil}
\newcommand{\rfrac}[1]{\left\lceil#1\right\rfloor}

\DeclareMathOperator{\union}{\bigcup}

\newcommand{\R}{\mathbb{R}}
\newcommand{\Z}{\mathbb{Z}}

\newcommand{\calA}{{\mathcal A}}
\newcommand{\calB}{{\mathcal B}}
\newcommand{\calC}{{\mathcal C}}
\newcommand{\calF}{{\mathcal F}}
\newcommand{\calJ}{{\mathcal J}}

\newcommand{\calS}{{\mathcal S}}

\newcommand{\calU}{{\mathcal U}}
\newcommand{\calV}{{\mathcal V}}

\newcommand{\bbT}{{\mathbb T}}

\newcommand{\eps}{\epsilon}

\newcommand{\dav}{d_{\mathsf {av}}}

\newcommand{\LP}{\mathsf{LP}}

\newcommand{\KM}{{\sf KM}\xspace}
\newcommand{\CKM}{{\sf CKM}\xspace}
\newcommand{\UFL}{{\sf UFL}\xspace}
\newcommand{\CFL}{{\sf CFL}\xspace}

\title{On Uniform Capacitated $k$-Median Beyond the Natural LP Relaxation}
\author{Shi Li  \\Toyota Technological Institute at Chicago \\ shili@ttic.edu}
\date{}

\begin{document}
\maketitle

\begin{abstract}
In this paper, we study the uniform capacitated $k$-median problem. In the problem, we are given a set $\calF$ of potential facility locations, a set $\calC$ of clients, a metric $d$ over $\calF \cup \calC$, an upper bound $k$ on the number of facilities we can open and an upper bound $u$ on the number of clients each facility can serve.  We need to open a subset $\calS \subseteq \calF$ of $k$ facilities and connect clients in $\calC$ to facilities in $\calS$ so that each facility is connected by at most $u$ clients. The goal is to minimize the total connection cost over all clients. Obtaining a constant approximation algorithm for this problem is a notorious open problem; most previous works gave constant approximations by either violating the capacity constraints or the cardinality constraint.  Notably, all these algorithms are based on the natural LP-relaxation for the problem. The LP-relaxation has unbounded integrality gap, even when we are allowed to violate the capacity constraints or the cardinality constraint by a factor of $2-\eps$. 

Our result is an $\exp(O(1/\eps^2))$-approximation algorithm for the problem that violates the cardinality constraint by a factor of $1+\eps$. That is, we find a solution that opens at most $(1+\eps)k$ facilities whose cost is at most $\exp(O(1/\eps^2))$ times the optimum solution when at most $k$ facilities can be open.  This is already beyond the capability of the natural LP relaxation, as it has unbounded integrality gap even if we are allowed to open $(2-\eps)k$ facilities.  Indeed, our result is based on a novel LP for this problem.  We hope that this LP is the first step towards a constant approximation for capacitated $k$-median.

The version as we described is the hard-capacitated version of the problem, as we can only open one facility at each location. This is as opposed to the soft-capacitated version, in which we are allowed to open more than one facilities at each location.  The hard-capacitated version is more general, since one can convert a soft-capacitated instance to a hard-capacitated instance by making enough copies of each facility location. We give a simple proof that in the uniform capacitated case, the soft-capacitated version and the hard-capacitated version are actually equivalent, up to a small constant loss in the approximation ratio. Moreover, we show that the given potential facility locations do not matter: we can assume $\calF = \calC$. 
\end{abstract}

\section{Introduction}
\label{section:introduction}

In the uniform capacitated $k$-median (\CKM) problem, we are given a set $\calF$ of potential facility locations, a set $\calC$ of clients, a  metric $d$ over $\calF \cup \calC$, an upper bound $k$ on the number of facilities we can open and an upper bound $u$ on the number of clients each facility can serve.  The goal is to find a set $\calS \subseteq \calF$ of at most $k$ open facilities and a connection assignment $\sigma : \calC \to \calS$ of clients to open facilities such that $\cardinal{\sigma^{-1}(i)} \leq u$ for every facility $i \in \calS$, so as to minimize the connection cost $\sum_{j \in \calC}d(j, \sigma(j))$.  

When $u = \infty$, the problem becomes the classical NP-hard $k$-median (\KM) problem. There has been extensive work on approximation algorithms for $k$-median.  The first constant approximation, due to Charikar et al.\ \cite{CGT99}, is an LP-based $6\frac23$-approximation. This factor was improved by a sequence of papers \cite{JV01, CG99, JMS02, AGK01, LS13}. 
In particular, Li and Svensson \cite{LS13} gave a $1+\sqrt{3} + \eps\approx 2.732 + \eps$-approximation for $k$-median, improving the previous decade-old ratio of $3+\eps$ due to \cite{AGK01}. Their algorithm is  based on a psudo-approximation algorithm that opens $k+O(1)$ facilities, and a process that turns a pseudo-approximation  into a true approximation.  Based on this framework, Byrka et al.\ \cite{BPR15} improved the approximation ratio from $2.732 + \eps$ to the current best $2.611+\eps$ very recently.  On the negative side, it is NP-hard to approximate the problem within a factor of $1+2/e-\eps\approx 1.736$ \cite{JMS02}. 

Little is known about the uniform \CKM problem; all constant approximation algorithms are pseudo-approximation algorithms, which produce solutions that violate either the capacity constraints or the cardinality constraint (the constraint that at most $k$ facilities are open). Charikar et al.\ \cite{CGT99} obtained a $16$-approximation for the problem, by violating the capacity constraint by a factor of $3$.  Later, Chuzhoy and Rabani \cite{CR05} gave a 40-approximation with capacity violation 50, for the more general \emph{non-uniform} capacitated $k$-median, where different facilities can have different capacities.  Recently, Byrka et al.\ \cite{BFR13} improved the capacity violation constant $3$ of \cite{CGT99} for uniform \CKM to $2+\eps$ and achieved approximation ratio of $O(1/\eps^2)$. This factor was improved to $O(1/\eps)$ by Li \cite{Li14}. Constant approximations for \CKM can also be achieved by violating the cardinality constraint. Gijswijt and Li \cite{GL13} designed a $(7+\eps)$-approximation algorithm for a more general version of \CKM that opens $2k+1$ facilities. 

There are two slightly different versions of the (uniform or non-uniform) \CKM problem. In the version as we described, we can open at most one facility at each location. This is sometimes called hard \CKM.  This is as opposed to soft \CKM, where we can open more than one facilities at each location.   Notice that hard \CKM is more general as one can convert a soft \CKM instance to a hard \CKM instance by making enough copies of each location. The result of Chuzhoy and Rabani \cite{CR05} is for soft \CKM while the other mentioned results are for (uniform or non-uniform) hard \CKM.

Most previous approximation algorithms on \CKM are based on the basic LP relaxation. A simple example shows that the LP has unbounded gap. This is the main barrier to a constant approximation for \CKM.  Moreover, the integrality gap is unbounded even if we are allowed to violate the cardinality constraint or the capacity constraint by a factor of $2-\eps$.  Thus, for algorithms based on the basic LP relaxation, \cite{Li14} and \cite{GL13} almost gave the smallest capacity violation factor and cardinality violation factor, respectively.

Closely related to \KM and \CKM are the uncapacitated facility location (\UFL) and capacitated facility location (\CFL) problems. \UFL has similar inputs as \KM but instead of giving an upper bound $k$ on the number of facilities we can open, it specifies an opening cost $f_i$ for each facility $i \in \calF$. The objective is the sum of the cost for opening facilities and the total connection cost. In \CFL, every facility $i \in \calF$ has a capacity $u_i$ on the maximum number of clients it can serve. There has been a steady stream of papers giving constant approximations for \UFL~\cite{LV92B,STA97,JV01,CS04,KPR98,CG99,JMM03,JMS02,MYZ06,Byr07}. The current best approximation ratio for \UFL is $1.488$ due to Li~\cite{Li11}, while the hardness of approximation is $1.463$~\cite{GK98}.

In contrast to \CKM, constant approximations are known for \CFL. Mahdian et al.\ \cite{MYZ06} gave a $2$-approximation for soft \CFL. For uniform hard \CFL, Korupolu et al.\ \cite{KPR98} gave an $(8+\eps)$-approximation, which was improved to $6+\eps$ by Chudak and Williamson \cite{CW05} and to $3$ by Aggarwal et al.\ \cite{AAB10}. For (non-uniform) hard \CFL, the best approximation ratio is $5$ due to Bansal et al.\ \cite{BGG12}, which improves the ratio of $3+2\sqrt{2}$ by Zhang et al.\ \cite{ZCY05}. All these algorithms for hard \CFL are based on local search. Recently, An et al.\ gave an LP-based constant approximation algorithm for hard \CFL \cite{ASS14}, solving a long-standing open problem \cite{WS11}.

\paragraph{Our contributions} In this paper, we introduce a novel LP for uniform \CKM, that we call the \emph{rectangle LP}. We give a rounding algorithm that achieves constant approximation for the problem, by only violating the cardinality constraint by a factor of $1+\eps$, for any constant $\eps > 0$.  This is already beyond the approximability of the basic LP relaxation, as it has unbounded integrality gap even if we are allowed to violate the cardinality constraint by $2-\eps$. To be more specific, we prove

\begin{theorem}\label{theorem:main}
Given a uniform capacitated $k$-median instance and a constant $\eps > 0$, we can find in polynomial time a solution with at most $\ceil{(1+\eps)k}$ open facilities and total connection cost at most $\exp(O(1/\eps^2))$ times the cost of the optimum solution with $k$ open facilities. 
\end{theorem}

The running time of our algorithm is $n^{O(1)}$, where the constant in the exponent does not depend on $\eps$.   If we allow the running time to be $n^{O(1/\eps)}$, we can remove the ceiling in the number of open facilities: we can handle the case when $k \leq O(1/\eps)$ by enumerating the $k$ open facilities.    As our LP overcomes the gap instance for the basic LP relaxation, we hope it is the first step towards a constant approximation for capacitated $k$-median. 

Our algorithm is for the hard capacitated version of the problem; namely, we open at most one facility at each location.  Indeed, we give a simple proof that, up to a constant loss in the approximation ratio, we can assume the instance is soft-capacitated and $\calF = \calC$.

\begin{theorem}
\label{theorem:soft-hard-same}
Let $(k, u, \calF, \calC, d)$ be a hard uniform \CKM instance,  and $C$ be the minimum connection cost of the instance when all facilities in $\calF$ are open.\footnote{Given the set of open facilities, finding the best connection assignment is a minimum cost bipartite matching problem.}  Then, given any solution of cost $C'$ to the soft uniform \CKM instance $(k, u, \calC, \calC, d)$, we can find a solution of cost at most $C + 2C'$ to the hard uniform \CKM instance $(k, u, \calF, \calC, d)$.
\end{theorem}

$C$ is a trivial lower bound on the cost of the hard uniform \CKM instance $(k, u, \calF, \calC, d)$. Moreover, the optimum cost of the soft uniform \CKM instance $(k, u, \calC, \calC, d)$ is at most twice the optimum cost of the hard uniform \CKM instance $(k, u, \calF, \calC, d)$.  Thus, any $\alpha$-approximation for the soft instance $(k, u, \calC, \calC, d)$ implies a $1 + 2(2\alpha) = (1+4\alpha)$-approximation for the hard instance $(k, u, \calF, \calC, d)$.  The reduction works even if we are considering pseudo-approximation algorithms by allowing violating the cardinality constraint by $\beta \geq 1$ and the capacity constraint by $\gamma \geq 1$; we can simply apply the above theorem to the instance $(\floor{\beta k}, \floor{\gamma u}, \calF, \calC, d)$.    Thus, we only focus on soft uniform \CKM instances with $\calF = \calC$ in the paper.  

Though we have $\calF = \calC$, we keep both notions to indicate whether a set of facility locations or a set of clients is being considered. Most part of our algorithm works without assuming $\calF = \calC$; only a single step uses this assumption.

The remaining part of the paper is organized as follows.  In Section~\ref{section:prelim}, we  introduce some useful notations, the basic LP relaxation for uniform \CKM, the gap instance and the proof of Theorem~\ref{theorem:soft-hard-same}. In Section~\ref{section:config-LP}, we describe our rectangle LP. Then in Section~\ref{section:rounding}, we show how to round a fractional solution obtained from the rectangle LP. 
\ifdefined\CR
\else
\fi
We leave some open questions in Section~\ref{section:discussion}.

\section{Preliminaries}
\label{section:prelim}
Let $\Z_+, \Z_*, \R_+$ and $\R_*$ denote the set of positive integers, non-negative integers, positive real numbers and non-negative real numbers respectively.  For any $x\in \R_*$, let $\floor{x}$ and $\ceil{x}$ denote the floor and ceiling of $x$ respectively. Let $\lfrac{x} = x - \floor{x}$ and $\rfrac{x} = \ceil{x} - x$.

Given two sets $\calC', \calC'' \subseteq \calC$ of points, define $d(\calC', \calC'') = \min_{j \in \calC', j' \in \calC''}d(j, j')$  be the minimum distance from points in $\calC'$ to points in $\calC''$. We simply use $d(j, \calC'')$ for $d(\set{j}, \calC'')$. 

Following is the basic LP for the uniform \CKM problem:
\ifdefined\CR
\vspace*{-0.5\abovedisplayskip}
\begin{equation}
\textstyle \min \qquad \sum_{i \in \calF, j\in\calC}d(i,j)x_{i,j} \qquad \text{s.t.} \tag{Basic LP}
\end{equation}
\vspace*{-20pt}
\begin{alignat}{2}\setlength{\abovedisplayskip}{0pt}\setlength{\belowdisplayskip}{0pt}
\textstyle \sum_{i \in \calF} y_i &\leq k,  \label{LPC:k-facilities} \\
\textstyle  \sum_{i \in \calF}x_{i,j} &=1, &\qquad &\forall j \in \calC, \label{LPC:client-must-connect} \\
\textstyle  x_{i,j} &\leq y_i, &\qquad &\forall i \in \calF, j \in \calC, \label{LPC:connect-to-open}\\
\textstyle  \sum_{j \in \calC}x_{i,j} &\leq uy_i, &\qquad &\forall i \in \calF, \label{LPC:capacity} \\
\textstyle  x_{i,j}, y_i &\geq 0, &\qquad &\forall i \in \calF, j \in \calC. \label{LPC:xy-non-neg} 
\end{alignat}
\else
\begin{equation}
\textstyle \min \qquad \sum_{i \in \calF, j\in\calC}d(i,j)x_{i,j} \qquad \text{s.t.} \tag{Basic LP}
\end{equation}
\vspace*{-30pt}

\begin{minipage}{0.4\textwidth}
\begin{alignat}{2}\setlength{\abovedisplayskip}{0pt}\setlength{\belowdisplayskip}{0pt}
\textstyle \sum_{i \in \calF} y_i &\leq k,  \label{LPC:k-facilities} \\
\textstyle  \sum_{i \in \calF}x_{i,j} &=1, &\qquad &\forall j \in \calC, \label{LPC:client-must-connect} \\
\textstyle  x_{i,j} &\leq y_i, &\qquad &\forall i \in \calF, j \in \calC, \label{LPC:connect-to-open}
\end{alignat}
\end{minipage}
\begin{minipage}{0.55\textwidth}
\begin{alignat}{2}\setlength{\abovedisplayskip}{0pt}\setlength{\belowdisplayskip}{0pt}
\textstyle  \sum_{j \in \calC}x_{i,j} &\leq uy_i, &\qquad &\forall i \in \calF, \label{LPC:capacity} \\
\textstyle  x_{i,j}, y_i &\geq 0, &\qquad &\forall i \in \calF, j \in \calC. \label{LPC:xy-non-neg} \\[8pt]
 \nonumber
\end{alignat}
\end{minipage}
\vspace*{1pt}
\fi

In the above LP, $y_i$ is the number of open facilities at location $i$, and $x_{i, j}$ indicates whether a client $j$ is connected to a facility at $i$.  Constraint~(\ref{LPC:k-facilities}) says that we can open at most $k$ facilities, Constraint~(\ref{LPC:client-must-connect}) says that every client must be connected to a facility, Constraint~(\ref{LPC:connect-to-open}) says that a client can only be connected to an open facility and Constraint~(\ref{LPC:capacity}) is the capacity constraint. In the integer programming capturing the problem, we require $y_i \in \Z_*$ and $x_{i, j} \in \set{0, 1}$ for every $i \in \calF, j \in \calC$. In the LP relaxation, we relax the constraint to $x_{i, j} \geq 0, y_i \geq 0$.

The basic LP has unbounded integrality gap, even if we are allowed to open $(2-\epsilon)k$ facilities.  The gap instance is the following.  $k = u+1$ and $|\calF| = |\calC| = n = u(u+1)$. The $n$ points are partitioned into $u$ groups, each containing $u+1$ points.  Two points in the same group have distance 0 and two points in different groups have distance 1. The following LP solution has cost 0: $y_i = 1/u$ for every $i \in \calF$ and $x_{i,j} $ is $1/(u+1)$ if $i$ is co-located with $j$ and $0$ otherwise.  The optimum solution is non-zero even if we are allowed to open $2u-1 = 2k-3$ facilities: there must be a group in which we open at most 1 facility and some client in the group must connect to a facility outside the group. \footnote{Note that this gap instance is not bad when we are allowed to violate the capacity constraints by $1+\eps$. However, if we are only allowed to violate the capacity constraints, there is a different bad instance: each group has $2u-1$ clients and $k = 2u-1$. Fractionally, we open $2-1/u$ facilities in each group and the cost is $0$. But if we want to open $2u-1$ facilities integrally, some group contains at most 1 facility and thus the capacity violation factor has to be $2-1/u$.}

\subsection{Reduction to Soft Capacitated Case: Proof of Theorem~\ref{theorem:soft-hard-same}} 
\label{subsec:soft-hard-same}
\begin{proof}
Consider the solution for the soft uniform capacitated \CKM instance $(k, u, \calC, \calC, d)$. We construct a set $\calS$ of size at most $k$ as follows. Suppose we opened $s$ facilities at some location $j \in \calC$, we add $s$ facility locations collocated with $j$ to $\calS$. By the assumption, we can find a matching of cost $C$ between $\calF$ and $\calC$ (the cost of matching $i\in \calF$ to $j \in \calC$ is $d(i, j)$), where each facility in $\calF$ is matched at most $u$ times and each client in $\calC$ is matched exactly once.   We are also given a matching of cost $C'$ between $\calC$ and $\calS$, where each client in $\calC$ is matched exactly once and each facility $i \in \calS$ is matched $t_i \leq u$ times.  By concatenating the two matchings and by triangle inequalities, we obtain a matching between $\calF$ and $\calS$ of cost at most $C+C'$, such that every facility in $\calF$ is matched at most $u$ times and every facility in $i \in \calS$ is matched $t_i$ times.  We then modify the matching between $\calF$ and $\calS$ in iterations, so that finally at most $\cardinal{\calS} \leq k$ facilities in $\calF$ are matched. Moreover, the modifications do not increase the cost of the matching. 

Focus on the bipartite multi-graph between $\calF$ and $\calS$ defined by the matching. Then we can assume the graph is a forest, when ignoring multiplicities.  If there is an even cycle, we can color the edges in the cycle alternatively in black and white. Assume the total length of black edges is at most that of white edges.  Then, we can increase the multiplicities of black edges by one and decrease the multiplicities of white edges by one.  This does not increase the cost of the matching. We can apply this operation until the cycle breaks.

We can further assume that in any tree of the forest, at most one facility in $\calF$ is matched less than $u$ times.  If there are two, we then take the path in the tree connecting the two facilities (path has even length), color the edges in the path alternatively in black and white. Assume the total length of black edges is at most that of white edges. Again we can increase the multiplicities of black edges and decrease the multiplicities of white edges.  We can apply this operation until either some edge disappears from the tree, or one of the two facilities is matched exactly $u$ times.

Now we claim that at most $k$ facilities in $\calF$ are matched.  To see this, focus on each tree in the forest containing at least one edge.  If facilities in $\calS$ in the tree are matched $t$ times in total,  so are the facilities in $\calF$ in the tree. Thus, there are exactly $\ceil{t/u}$ facilities in $\calF$ in this tree, since at most one facility in $\calF$ in the tree is matched less than $u$ times. The number of facilities in $\calS$ in this tree is at least $\ceil{t/u}$ since each facility in $\calS$ is matched $t_i \leq u$ times.  This proves the claim. 

Let $\calF' \subseteq \calF$ be the set of facilities that are matched.  Then, $\cardinal{\calF'} \leq \cardinal{\calS} \leq k$, and we have a matching between $\calF'$ and $\calS$ of cost at most $C+C'$, where each facility in $\calF'$ is matched at most $u$ times and each facility in $\calS$ is matched $t_i$ times. By concatenating this matching with the matching between $\calS$ and $\calC$ of cost $C'$, we obtain a solution of cost $C+2C'$ with open facilities $\calF'$ to the uniform hard \CKM instance $(k, u, \calF, \calC, d)$. This finishes the proof. 
\end{proof}
\section{Rectangle LP}
\label{section:config-LP}

Our rectangle LP is motivated by the gap instance described in Section~\ref{section:prelim}.  Focus on a group of $u+1$ clients in the gap instance.  The fractional solution opens $1+1/u$ facilities for this group and use them to serve the $u(1+1/u) = u+1$ clients in the group.   We interpret this fractional event as a convex combination of integral events: with probability $1-1/u$ we open 1 facility for the group and serve $u$ clients; with probability $1/u$ we open 2 facilities and serve $2u$ clients.  However, there are only $u+1$ clients in this group; even if $2$ facilities are open, we can only serve $u+1$ clients.  Thus, we can only serve $(1-1/u)u + (1/u) (u+1) = u + 1/u < u+1$ clients using $1+1/u$ open facilities. 

This motivates the following definition of  $f(p, q)$ for any $p \in \Z_*, q \in \R_*$.  When $q \in \Z_*$,  let $f(p, q) = \min \set{qu, p}$ be the upper bound on the number of clients in a set of cardinality $p$ that can be connected to a set of $q$ facilities. We then extend the range of $q$ from $\Z_*$ to $\R_*$ using linear interpolation(see Figure~\ref{figure:f}).  Then the exact definition of $f(p, q)$  is the following:
\begin{align*}
f(p, q) = \begin{cases}
qu & q \leq \floor{ \frac pu}\\
u\floor{\frac pu} + u\lfrac{\frac pu}\left(q - \floor{\frac pu}\right) & \floor{\frac pu} < q < \ceil{\frac pu}\\
p & q \geq \ceil{\frac pu}
\end{cases}.
\end{align*}

\begin{figure}
\centering
\includegraphics[width=0.45\textwidth]{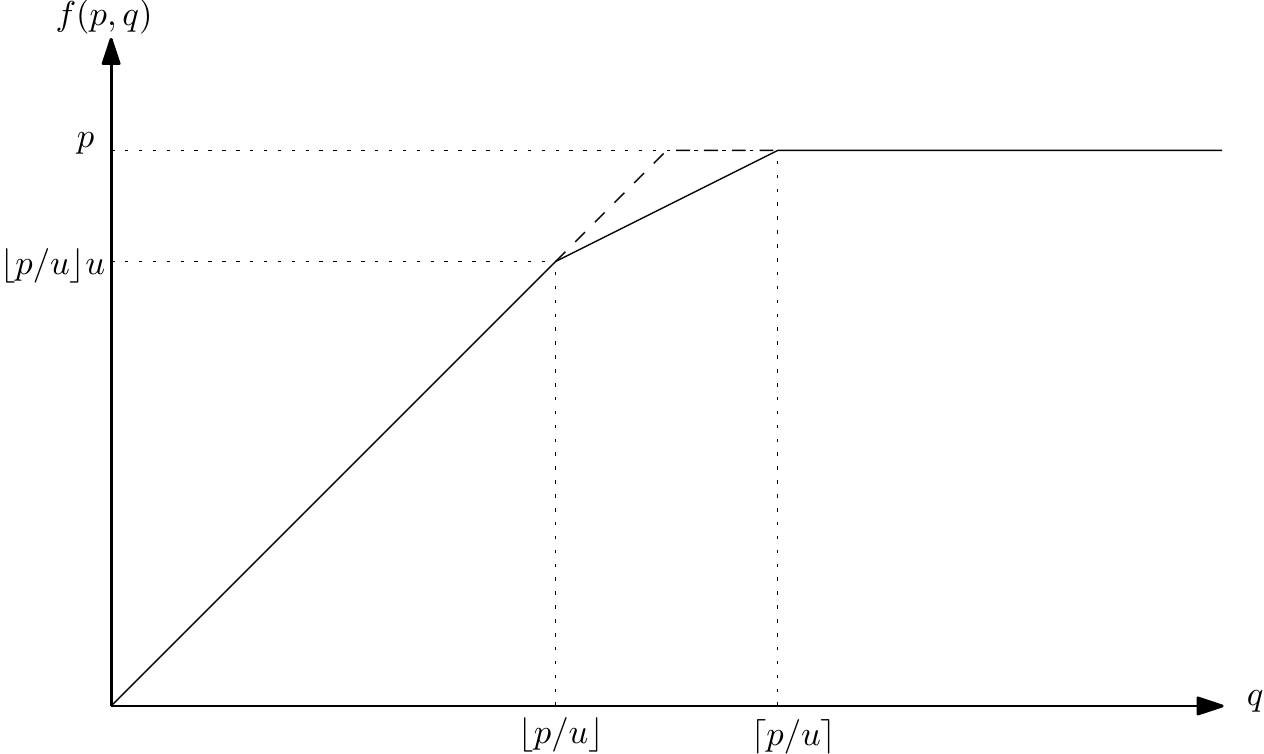}
\caption{The $f$ function for a fixed $p$ such that $p/u \notin \Z$. It contains three linear segments. The dashed line shows the function $f$ defined by $f = \min\set{p, qu}$.}
\label{figure:f}
\end{figure}

\begin{claim}
\label{lemma:concave}
Fixing $p \in \Z_*$, $f(p, \cdot)$ is a concave function on $\R_*$. Fixing $q \in \R_*$, $f(\cdot, q)$ is a concave function on $\Z_*$.
\end{claim}

\begin{proof}
It is easy to see that $f(p, q) = \min\{p, uq, \allowbreak u \floor{p/u} + u\lfrac{p/u}(q - \floor{p/u})\}$. Fix $p$, all the three terms are linear functions of $q$; thus the minimum of the three is concave. 

Now we fix $q \in \R_*$. Then $f(p, q)=p$ if $p \leq u\floor{q}$, $f(p, q) = u\floor{q} + (p-u\floor{q})\lfrac{q}$ if $u\floor{q} < p < u\ceil{q}$, and $f(p, q) = uq$ if $p \geq u\ceil{q}$. All three segments are linear on $p$ and their gradients are $1, \lfrac{q}, 0$ respectively. The gradients are decreasing from left to right. Moreover, the first segment and the second segment agree on $p = u\floor{q}$; the second segment and the third segment agree on $p = u\ceil{q}$. Thus, $f(\cdot, q)$ is a concave function on $\Z_*$. 
\end{proof}

For any subset $\calB \subseteq \calF$ of facility locations and subset $\calJ \subseteq \calC$ of clients, define $y_{\calB} := y(\calB):= \sum_{i \in \calB} y_i$ and $x_{\calB, \calJ} = \sum_{i \in \calB, j \in \calJ}x_{i,j}$. We simply write $x_{i, \calJ}$ for $x_{\set{i}, \calJ}$ and $x_{\calB, j}$ for $x_{\calB, \set{j}}$.  By the definition of $f(p, q)$,  $\sum_{j \in \calJ}x_{\calB, j} \leq f(\cardinal{\calJ}, y_\calB)$ is valid for every $\calB \subseteq \calF$ and  $\calJ \subseteq \calC$. The constraint says that there can be at most $f(\cardinal{\calJ}, y_\calB)$ clients in $\calJ$ connected to facilities in $\calB$.   Notice the constraint with $\calB = \set{i}$ and $\calJ = \set{j}$ implies $x_{i,j} \leq f(1, y_i) \leq y_i$. The constraint with $\calB = \set{i}$ and $\calJ = \calC$ implies $\sum_{j \in \calC}x_{i,j} \leq f(\cardinal{\calC}, y_i) \leq uy_i$.  Thus,   Constraint~\eqref{LPC:connect-to-open}  and~\eqref{LPC:capacity} are implied.  The constraints of our rectangle LP are Constraint~\eqref{LPC:k-facilities},\eqref{LPC:client-must-connect},\eqref{LPC:xy-non-neg} and the new constraint:
\ifdefined \CR
\begin{align}
&\text{minimize }  \sum_{i \in \calF, j \in \calC}x_{i,j}d(i,j) \text{ s.t.} \tag{Rectangle LP}
\end{align}
\vspace*{-\abovedisplayskip}
\vspace*{-\belowdisplayskip}
\begin{align}
&x_{\calF,j} = 1,y_\calF \leq k, x_{i, j}, y_i \geq 0,&\  \forall i \in \calF, j \in \calC,
 \label{LPC:compact-simple} \\ 
 &x_{\calB, \calJ} \leq f(\cardinal{\calJ}, y_\calB),&\ \forall  \calB \subseteq \calF, \calJ \subseteq \calC.
 \label{LPC:compact}
\end{align}
\else
\begin{alignat}{2}
\text{minimize } &\sum_{i \in \calF, j \in \calC}x_{i,j}d(i,j) \text{ s.t.} \tag{Rectangle LP}\\
&x_{\calF,j} = 1,\  y_\calF \leq k,\ x_{i, j} \geq 0,\  y_i \geq 0, &\qquad \forall i \in \calF, j \in \calC, \label{LPC:compact-simple} \\ 
&x_{\calB, \calJ} \leq f(\cardinal{\calJ}, y_\calB), &\qquad \forall  \calB \subseteq \calF, \calJ \subseteq \calC. \label{LPC:compact}
\end{alignat}
\fi

The LP is called the rectangle LP since we have a constraint for every ``rectangle'' $(\calB \subseteq \calF, \calJ \subseteq \calC)$. We use the concavity of $f(p, \cdot)$ to convert Constraint~\eqref{LPC:compact} to linear constraints. Since $f(p, q)$ is the minimum of $p, uq$ and $u\floor{p/u} + u\lfrac{p/u}\left(q - \floor{p/u}\right)$, Constraint~\eqref{LPC:compact} is equivalent to a combination of three linear constraints. 

For a fixed $\calB \subseteq \calF$, the separation oracle for Constraint~\eqref{LPC:compact} is simple: for every $p \in [\cardinal{\calC}]$, we take the sum of the $p$ largest values in $\set{x_{\calB, j}: j \in \calC}$; if it is larger than $f(p, y_\calB)$ we find a separation.   Since there are exponential number of sets $\calB$, we do not know how to find a separation oracle for the Constraint~\eqref{LPC:compact} efficiently.  However, we can use the following standard trick: given $\set{x_{i,j} : i \in \calF, j \in \calC}$ and $\set{y_i : i \in \calF}$ satisfying Constraint~\eqref{LPC:compact-simple}, we either find a  rectangle $(\calB\subseteq \calF, \calJ \subseteq \calC)$ for which Constraint~\eqref{LPC:compact} is violated, or construct an integral solution with at most $\ceil{(1+\eps)k}$ facilities and the desired approximation ratio.  This is sufficient for us to run the ellipsoid method. 

We also remark that the definition of $f(p, q)$ for $\floor{p/u} < q < \ceil{p/u}$ is what makes the rectangle LP powerful. If we change the definition of $f(p, q)$ to $f(p, q) = \min\set{p, uq}$(see Figure~\ref{figure:f}), then the rectangle LP is equivalent to the basic LP.
\section{Rounding a Fractional Solution of the Rectangle LP}
\label{section:rounding}

Throughout this section,  let $\big(\set{x_{i,j}: i \in \calF, j \in \calC}, \allowbreak \set{y_i:i \in \calF}\big)$ be a fractional solution satisfying Constraints~\eqref{LPC:compact-simple}.  Let $\LP := \sum_{i \in \calF, j \in \calC}x_{i,j}d(i,j)$ be the cost of the fractional solution. We then try to round the fractional solution to an integral one with at most $\ceil{(1+\eps)k}$ open facilities.  We either claim the constructed integral solution has connection cost at most $\exp(O(1/\eps^2))\LP$, or output a rectangle $(\calB \subseteq \calF, \calJ \subseteq \calC)$ for which Constraint~\eqref{LPC:compact} is violated.   We can assume Constraint~\eqref{LPC:connect-to-open} and~\eqref{LPC:capacity} are satisfied by checking Constraint~\eqref{LPC:compact} for rectangles $(\set{i}, \set{j})$ and $(\set{i}, \calC)$ respectively.

Overall, the algorithm works as follows. Initially, we have 1 unit of demand at each client $j \in \calC$. During the execution of the algorithm, we move demands fractionally between clients. We pay a cost of $xd(j, j')$ for moving $x$ units of demand from client $j$ to client $j'$.  Suppose our final moving cost is $C$, and each client $j \in \calC$ has $\alpha_j$ units of demand.  Then we use the fact that $\calF = \calC$. We open $\ceil{\alpha_j/u}$ facilities at the location $j \in \calC = \calF$. By the integrality of matching, there is an integral matching between the $\calF$ and $\calC$, such that each $i \in \calF$ is matched at most $u\ceil{\alpha_i/u}$ times and each $j \in \calC$ is matched exactly once. The cost of the matching is at most $C$ (cost of matching $i$ and $j$ is $d(i, j)$). Thus our goal is to bound $C$ and $\sum_{j \in \calC}\ceil{\alpha_j/u}$.

\subsection{Moving Demands to Client Representatives}
In this section, we define a subset of clients called \emph{client representatives} (representatives for short) and move all demands to the representatives. The definition of client representatives is similar to that of Charikar and Li \cite{CL12}. 

Let $\dav(j) = \sum_{i \in \calF}x_{i, j}d(i, j)$ be the connection cost of $j$, for every client $j \in \calC$. Then $\LP = \sum_{j \in \calC}\dav(j)$.  Let $\ell = \Theta(1/\epsilon)$ be an integer whose value will be decided later.  Let $\calC^* = \emptyset$ initially. Repeat the following process until $\calC$ becomes empty. We select the client $v \in \calC$ with the smallest $\dav(v)$ and add it to $\calC^*$. We remove all clients $j$ such that $d(j, v) \leq 2\ell \dav(j)$ from $\calC$ (thus, $v$ itself is removed). Then the final set $\calC^*$ is the set of client representatives. We shall use $v$ and its derivatives to index representatives, and $j$ and its derivatives to index general clients. 

We partition the set $\calF$ of locations according to their nearest representatives in $\calC^*$.  Let $\calU_v = \emptyset$ for every $v \in \calC^*$ initially. For each location $i \in \calF$, we add $i$ to $\calU_v$ for the $v \in \calC^*$ that is closest to $i$. Thus, $\set{\calU_v:v \in \calC^*}$ forms a Voronoi diagram of $\calF$ with centers being $\calC^*$. For any subset $\calA \subseteq \calC^*$ of representatives, we use $\calU_{\calA} = \union_{v \in \calA} \calU_v$ to denote the union of Voronoi regions with centers in $\calA$.

\begin{claim} \label{claim:cluster-centers}
The following statements hold:
\begin{enumerate}[label=(C\arabic*),labelindent=*, leftmargin=*]
\item for all $v, v' \in \calC^*, v \neq v'$, we have $d(v, v') > 2\ell \max\set{\dav(v), \dav(v')}$; \label{property:representatives-far-away}
\item for all $j \in \calC$, there exists $v \in \calC^*$, such that $\dav(v) \leq \dav(j)$ and $d(v, j) \leq 2\ell\dav(j)$; \label{property:near-a-representative}
\item $y(\calU_v) \geq 1-1/\ell$ for every $v \in \calC^*$; \label{property:bundle-large} 
\item for any $v \in \calC^*$, $i \in \calU_v$ and $j \in \calC$, we have $d(i, v) \leq d(i, j) + 2\ell\dav(j)$. \label{property:moving-to-centers}
\end{enumerate}
\end{claim}

\begin{proof}
First consider Property~\ref{property:representatives-far-away}.  Assume $\dav(v) \leq \dav(v')$. When we add $v$ to $\calC^*$, we remove all clients $j$ satisfying $d(v, j) \leq 2\ell\dav(j)$ from $\calC$. Thus, $v'$ can not be added to $\calC^*$ later.

For Property~\ref{property:near-a-representative}, just consider the iteration in which $j$ is removed from $\calC$.  The representative $v$ added to $\calC^*$ in the iteration satisfy the property.

Then consider Property~\ref{property:bundle-large}. By Property~\ref{property:representatives-far-away}, we have $\calB:=\set{i \in \calF:d(i, v) \leq \ell\dav(v)} \subseteq \calU_v$. Since $\dav(v)=\sum_{i \in \calF}x_{i,v}d(i,v)$ and $\sum_{i \in \calF}x_{i,v} = 1$, we have $\dav(v) \geq (1-x_{\calB, v})\ell\dav(v)$, implying $y(\calU_v) \geq y_\calB \geq x_{\calB, v} \geq 1-\frac1\ell$, due to Constraint~\eqref{LPC:connect-to-open}.

Finally, consider Property~\ref{property:moving-to-centers}. By Property~\ref{property:near-a-representative}, there is a client $v' \in \calC^*$ such that $\dav(v') \leq \dav(j)$  and $d(v', j) \leq 2\ell \dav(j)$. Notice that $d(i, v) \leq d(i, v')$ since $v' \in \calC^*$ and $i$ was added to $\calU_v$. Thus, $d(i, v) \leq d(i, v') \leq d(i, j) + d(j, v') \leq d(i, j) + 2\ell\dav(j)$.
\end{proof}

Now, we move demands to $\calC^*$. For every representative $v \in \calC^*$, every location $i \in \calU_v$ and every client $j \neq v$ such that $x_{i,j} > 0$, we move $x_{i,j}$ units of demand from $j$ to $v$. We bound the moving cost:

\begin{lemma}\label{lemma:moving-to-centers}
The total cost of moving demands in the above step is at most $2(\ell + 1)\LP$.
\end{lemma}

\begin{proof}
The cost is bounded by
\ifdefined \CR
\begin{align*}
&\quad\sum_{v \in \calC^*}\sum_{i \in \calU_v}\sum_{j \in \calC}x_{i,j} (d(j, i) + d(i, v))\\
&\leq \sum_{v \in \calC^*}\sum_{i \in \calU_v}\sum_{j \in \calC}x_{i,j} (2d(i, j) + 2\ell\dav(j))\\
&= 2\sum_{j \in \calC} \sum_{v \in \calC^*, i \in \calU_v} x_{i,j}\left(d(i,j) + \ell\dav(j)\right)\\ &=2\sum_{j \in \calC} \left(\dav(j) + \ell\dav(j)\right)= 2(\ell+1)\LP.
\end{align*}
\else
\begin{align*}
&\quad\sum_{v \in \calC^*}\sum_{i \in \calU_v}\sum_{j \in \calC}x_{i,j} (d(j, i) + d(i, v)) \leq \sum_{v \in \calC^*}\sum_{i \in \calU_v}\sum_{j \in \calC}x_{i,j} (2d(i, j) + 2\ell\dav(j))\\
&= 2\sum_{j \in \calC} \sum_{v \in \calC^*, i \in \calU_v} x_{i,j}\left(d(i,j) + \ell\dav(j)\right) =2\sum_{j \in \calC} \left(\dav(j) + \ell\dav(j)\right)= 2(\ell+1)\LP.
\end{align*}
\fi
The inequality is by Property~\ref{property:moving-to-centers}. The second equality used the fact that $\set{\calU_v : v \in \calC^*}$ form a partition of $\calF$,  $\sum_{i \in \calF}x_{i,j} = 1$ and $\sum_{i \in \calF}x_{i,j}d(i, j) = \dav(j)$.
\end{proof}

After the moving operation, all demands are at the set $\calC^*$ of representatives.  Every representative $v \in \calC^*$ has $\sum_{i \in \calU_v}\sum_{j \in \calC}x_{i,j}$ units of demand. Let $y'_i := \frac{\sum_{j \in \calC}x_{i,j}}{u}$ for any facility location $i \in \calF$.  Since Constraint~\eqref{LPC:capacity} holds, we have $y'_i \leq y_i$. Define $y'_{\calB}:=y'(\calB):=\sum_{i \in \calB}y'_i = \frac{\sum_{j \in {\calC}}x_{\calB, j}}{u}$ for every $\calB \subseteq \calF$. Obviously $y'_\calB \leq y_\calB$. The amount of demand at $v \in \calC^*$ is $\sum_{i \in \calU_v}uy'_i = uy'(\calU_v)$.   

We have obtained an $O(1)$ approximation with $2k$ open facilities: we set $\ell = 2$ and open $\ceil{y'(\calU_v)}$ facilities at each location $v \in \calC^* \subseteq \calC = \calF$.  By Lemma~\ref{lemma:moving-to-centers}, the connection cost is at most $2(\ell + 1)\LP = 6\LP$. The number of open facilities is at most $2k$, as $\max_{v \in \calC^*} \frac{\ceil{y'(\calU_v)}}{y(\calU_v)} \leq \max_{y \geq 1-1/\ell}\frac{\ceil{y}}{y} \leq 2$. No matter how large $\ell$ is, the  bound is tight as $\frac{\ceil{1+\eps}}{1+\eps}$ approaches 2. This is as expected since we have not used Constraint~\eqref{LPC:compact}. In order to improve the factor of $2$, we further move demands between client representatives. 

\subsection{Bounding cost for moving demands out of a set}
Suppose we are given a set $\calA \subseteq \calC^*$ of representatives such that $d(\calA, \calC^* \setminus \calA)$ is large. If $\ceil{y'(\calU_\calA)}/y(\calU_\calA)$ is large then we can not afford to open $\ceil{y'(\calU_\calA)}$ open facilities inside $\calA$. (Recall that $\calU_\calA = \union_{v \in \calA}\calU_v$ is the union of Voronoi regions with centers in $\calA$.)  Thus, we need to move demands between $\calA$ and $\calC^* \setminus \calA$. The goal of this section is to bound $d(\calA, \calC^* \setminus \calA)$; this requires Constraint~\eqref{LPC:compact}.

To describe the main lemma, we need some notations.  Let $D_i = \sum_{j \in \calC}x_{i,j}d(i, j)$ and $D'_i = \sum_{j \in \calC}x_{i,j}\dav(j)$ for any location $i \in \calF$. Let $D_{\calF'} := D(\calF') := \sum_{i \in \calF'}D_i$ and $D'_{\calF'} := D'(\calF') := \sum_{i \in \calF'}D'_i$ for every subset $\calF' \subseteq \calF$ of locations. It is easy to see that $\LP = D_\calF = D'_\calF$; this fact will be used to bound the total moving cost.  The main lemma we prove in this section is the following.
\begin{lemma}
\label{lemma:bounding-moving-cost}
Let $\emptyset \subsetneq {\calA} \subsetneq \calC^*$ and $\calS = \calU_\calA$.   Suppose $y'_\calS \geq \floor{y_\calS}$ and Constraint~\eqref{LPC:compact} holds for $\calB = \calS$ and every $\calJ \subseteq \calC$. Then, 
\begin{align*}
\lfrac{y'_\calS}\rfrac{y_\calS} d(\calA, \calC^* \setminus \calA) \leq \frac{4}{u}D_\calS + \frac{4\ell + 2}{u}D'_\calS.
\end{align*}
\end{lemma}

We explain why this bound gives what we need. We can open $\floor{y'_\calS}$ facilities in $\calA$ and move $u\lfrac{y'_\calS}$ units of demand from $\calA$ to some close representatives in $\calC^* \setminus \calA$. If we guarantee that the moving distance is roughly $d(\calA, \calC^* \setminus \calA)$, then the moving cost is $u\lfrac{y'_\calS}d(\calA, \calC^* \setminus \calA)$. When $\rfrac{y_\calS}$ is not too small, the cost is bounded in terms of $D_\calS + D'_\calS$. On the other hand, if $\rfrac{y_\calS}$ is very small, we can simply open $\ceil{y_\calS}$ facilities in $\calA$ as $\ceil{y_\calS}/y_\calS$ is close to 1. 

The proof of Lemma~\ref{lemma:bounding-moving-cost} requires the following lemma, which directly uses the power of Constraint~\eqref{LPC:compact}. As the lemma is very technical, we defer its proof to Section~\ref{subsec:omitted-proofs}.  We shall prove Lemma~\ref{lemma:bounding-moving-cost} assuming Lemma~\ref{lemma:xb-times-xbb-large}.

\begin{lemma}\label{lemma:xb-times-xbb-large}
Suppose $(\set{x_{i,j}:i \in \calF, j \in \calC}, \set{y_i:i \in \calF})$ satisfies Constraint~\eqref{LPC:compact} for some set $\calB \subseteq \calF$ and every $\calJ \subseteq \calC$. Moreover, suppose $y'_\calB \geq \floor{y_\calB}$. Then
\begin{align}
\sum_{j \in \calC}x_{\calB,j}(1-x_{\calB,j}) \geq u\lfrac{y'_\calB}\rfrac{y_\calB}. \label{equ:xb-times-xbb-large}
\end{align}
\end{lemma}

To get an intuition about Inequality~\eqref{equ:xb-times-xbb-large}, let us assume $y'_\calB = y_\calB \neq \Z$ and $uy_\calB \in \Z$. Thus, $\calB$ serves $uy'_\calB = uy_\calB$ fractional clients. Without Constraint~\eqref{LPC:compact}, it can happen that $B$ serves $uy_\calB$ integral clients, in which case the left side of \eqref{equ:xb-times-xbb-large} is $0$ and \eqref{equ:xb-times-xbb-large} does not hold. In other words, Inequality~\eqref{equ:xb-times-xbb-large} prevents the case from happening. Indeed, we show that the left side of \eqref{equ:xb-times-xbb-large} is minimized when the following happens: $\calB$ serves $u\floor{y_\calB}$ integral clients, and $u$ fractional clients, each with fraction $\lfrac{y_\calB}$. In this case, \eqref{equ:xb-times-xbb-large} holds with equality. 

\ifdefined\CR
\smallskip
\begin{proofof}{\bf Proof of Lemma~\ref{lemma:bounding-moving-cost}}
\else
\begin{proof}[\bf Proof of Lemma~\ref{lemma:bounding-moving-cost}]
\fi
Focus on some $i \in \calS, i' \in \calF \setminus \calS, j \in \calC$. Suppose $i \in \calU_{v}$ for some $v \in {\calA}$ and $i' \in \calU_{v'}$ for some $v' \in \calC^* \setminus \calA$.  Then 
\ifdefined\CR
\begin{align*}
&\quad d(\calA, \calC^* \setminus \calA) \quad \leq \quad d(v, v') \quad \leq \quad d(v, i') + d(v', i') \\
&\leq \quad 2d(v, i') \quad \leq \quad 2(d(v, i) + d(i, j) + d(j, i')) \\ 
&\leq \quad 2[2d(i,j)+2\ell\dav(j) + d(i', j)]. 
\end{align*}
\else
\begin{align*}
d(\calA, \calC^* \setminus \calA) & \quad \leq \quad d(v, v') \quad \leq \quad d(v, i') + d(v', i') \quad \leq \quad 2d(v, i') \quad \leq \quad 2(d(v, i) + d(i, j) + d(j, i')) \\ 
& \quad \leq \quad 2[2d(i,j)+2\ell\dav(j) + d(i', j)]. 
\end{align*}
\fi
In the above sequence, the third inequality used the fact that $i' \in \calU_{v'}$ and the fifth inequality used Property~\ref{property:moving-to-centers} in Claim~\ref{claim:cluster-centers}. Thus, 
\ifdefined\CR
\begin{align*}
&\quad\ \lfrac{y'_\calS}\rfrac{y_\calS} d(\calA, \calC^* \setminus \calA) \\
&\leq \ \frac{1}{u}\sum_{j \in \calC}x_{\calS,j}(1-x_{\calS,j})d(\calA, \calC^* \setminus \calA)\\
&=\ \frac1u\sum_{j\in \calC, i \in \calS, i' \in \calF\setminus \calS} x_{i,j}x_{i',j} d(\calA, \calC^* \setminus \calA) \\
&\leq \  \frac{2}{u}\sum_{j,i,i'}x_{i,j}x_{i',j}\left[2d(i,j) + 2\ell\dav(j) + d(i', j)\right]\\
&=\  \frac{4}{u}\sum_{j, i}x_{i,j}(1-x_{\calS,j})\left[d(i,j) + \ell\dav(j)\right]  \\
&\qquad \qquad+ \frac{2}{u}\sum_{j, i'}x_{\calS,j}x_{i',j}d(i',j)\\
&\leq\  \frac{4}{u}\sum_{i, j}x_{i,j}\left[d(i,j) + \ell\dav(j)\right] + \frac{2}{u}\sum_{j}x_{\calS,j}\dav(j)\\
&=\ \frac{4}{u}D_\calS + \frac{4\ell + 2}{u}D'_\calS.
\end{align*}
\else
\begin{align*}
&\quad\ \lfrac{y'_\calS}\rfrac{y_\calS} d(\calA, \calC^* \setminus \calA)  \leq \ \frac{1}{u}\sum_{j \in \calC}x_{\calS,j}(1-x_{\calS,j})d(\calA, \calC^* \setminus \calA)\\
&=\ \frac1u\sum_{j\in \calC, i \in \calS, i' \in \calF\setminus \calS} x_{i,j}x_{i',j} d(\calA, \calC^* \setminus \calA)  \leq \  \frac{2}{u}\sum_{j,i,i'}x_{i,j}x_{i',j}\left[2d(i,j) + 2\ell\dav(j) + d(i', j)\right]\\
&=\  \frac{4}{u}\sum_{j, i}x_{i,j}(1-x_{\calS,j})\left[d(i,j) + \ell\dav(j)\right]  + \frac{2}{u}\sum_{j, i'}x_{\calS,j}x_{i',j}d(i',j)\\
&\leq\  \frac{4}{u}\sum_{i, j}x_{i,j}\left[d(i,j) + \ell\dav(j)\right] + \frac{2}{u}\sum_{j}x_{\calS,j}\dav(j) =\ \frac{4}{u}D_\calS + \frac{4\ell + 2}{u}D'_\calS.
\end{align*}
\fi
In above summations, $j$ is over all clients in $\calC$, $i$ is over all locations in $\calS$ and $i'$ is over all locations in $\calF \setminus \calS$.   The first inequality in the sequence used Lemma~\ref{lemma:xb-times-xbb-large}. All other inequalities and equations follow from the definitions of the notations used. 
\ifdefined\CR
\end{proofof}
\else
\end{proof}
\fi

\subsection{Constructing family of neighborhood trees}
Lemma~\ref{lemma:bounding-moving-cost} gives a necessary bound for our analysis. Still, we need to guarantee some other conditions when moving the demands. For example, when moving demands out of an ``isolated'' set $\calA$, we should make sure that the distance is roughly $d(\calA, \calC^* \setminus \calA)$. If $y'(\calU_\calA) \leq \floor{y(\calU_\calA)}$, then we should not move demands out of $\calA$, as $d(\calA, \calC^* \setminus \calA)$ may not be bounded any more. 

We guarantee these conditions by building a set of rooted trees over $\calC^*$, called \emph{neighborhood trees}. Roughly speaking, each neighborhood tree contains representatives that are nearby; moving demands within a neighborhood tree does not cost too much. 

We use a triple $T = ({\calV}, E, r)$ to denote a rooted tree, with vertex set ${\calV} \subseteq \calC^*$, edge set $E \subseteq {{\calV} \choose 2}$ and root $r \in {\calV}$.   Given a rooted tree $T = ({\calV}, E, r)$ and a vertex $v \in {\calV}$, we use $\Lambda_T(v)$ to denote the set of vertices in the sub-tree of $T$ rooted at $v$. If $v \neq r$, we use $\rho_T(v)$ to denote the parent of $v$ in $T$.

\begin{Definition}
A rooted tree $T = ({\calV} \subseteq \calC^*, E, r)$ is called a \emph{neighborhood tree} if for every vertex $v \in {\calV} \setminus r$, $d(v, \calC^* \setminus \Lambda_T(v)) = d(v, \rho_T(v))$. 
\end{Definition}

In other words, $T=({\calV}, E, r)$ is a neighborhood tree if for every non-root vertex $v$ of $T$, the parent $\rho_T(v)$ of $v$ is the nearest vertex in $\calC^*$ to $v$, except for $v$ itself and its descendants.  
The next lemma shows that we can cover $\calC^*$ using a set of neighborhood trees of size between $\ell$ and $\ell^2$.  The vertex sets of these trees almost form a partition of $\calC^*$, except that trees may share the same root. Since the lemma is technical and peripheral to the spirit of our result, we defer the proof to Section~\ref{subsec:omitted-proofs}.

\begin{lemma} \label{lemma:constructing-neighborhood-trees}
Given any positive integer $\ell$ such that $\cardinal{\calC^*} \geq \ell$, we can find a set $\bbT$ of neighborhood trees such that 
\begin{enumerate}[label=(T\arabic*),labelindent=*, leftmargin=*]
\item $\ell \leq \cardinal{{\calV}} \leq \ell^2$ for every neighborhood-tree $({\calV}, E, r) \in \bbT$; \label{property:trees-first} 
\label{property:trees-bounded-size}
\item $\union_{({\calV}, E, r) \in \bbT} {\calV} = \calC^*$;  \label{property:trees-covering}
\item For two distinct trees $T = ({\calV}, E, r), T' = ({\calV}', E', r') \in \bbT$, ${\calV}\setminus \set{r}$ and ${\calV}' \setminus \set{r'}$ are disjoint.  
\label{property:trees-last} 
\label{property:trees-partitioning}
\end{enumerate}
\end{lemma}

\subsection{Moving demands within neighbourhood trees}
Recall that all the demands are at the client representatives. Every representative $v \in \calC^*$ has $uy'(\calU_v)$ units of demand. \emph {In this section, it is convenient for us to {\bf scale down} the demands by $u$.  Thus a representative $v \in \calC^*$ has $y'(\calU_v)$ units of demand.  Due to the scaling, moving $x$ units of demand from $v$ to $v'$ costs $uxd(v, v')$. If finally some $v$ has $\alpha_v$ units of demand, we need to open $\ceil{\alpha_v}$ facilities at $v$.} For analytical purposes, we also say that $v \in \calC^*$ has $y(\calU_v) \geq y'(\calU_v)$ units of supply.  The total supply is $\sum_{v \in \calC^*}y(\calU_v) = y_\calF \leq k$. 

Assume $\cardinal{\calC^*} \geq \ell$ for now. We apply Lemma~\ref{lemma:constructing-neighborhood-trees} to construct a set $\bbT$ of neighborhood trees satisfying Properties~\ref{property:trees-first} to \ref{property:trees-last}.  We assign the supplies and demands to vertices in the set $\bbT$. Notice that every representative in $\calC^*$ appears in $\bbT$, and it appears in $\bbT$ as a non-root at most once.  If $v \in \calC^*$ appears as a non-root, we assign the $y'(\calU_v)$ units of demand and the $y(\calU_v)$ units of supply to the non-root. Otherwise, we assign the $y'(\calU_v)$ units of demand and the $y(\calU_v)$ units of supply to an arbitrary root $v$ in $\bbT$.

Fix a neighborhood tree $T=(\calV, E, r) \in \bbT$ from now on.   Each $v \in \calV$ has $\alpha_v$ units of demand and $\beta_v$ units of supply. For $v \in \calV \setminus \set{r}$, we have $\alpha_v = y'(\calU_v)$ and $\beta_v = y(\calU_v)$. We have either $\alpha_r = y'(\calU_r), \beta_r = y(\calU_r)$ or $\alpha_r = \beta_r = 0$. Define $\alpha_{{\calV}'} := \sum_{v \in {\calV}'}\alpha_v$ and $\beta_{{\calV}'} :=\sum_{v \in {\calV}'}\beta_v$ for every ${\calV}' \subseteq {\calV}$.  We shall move demands and supplies within $T$.  Moving supplies is only for analytic purposes and costs nothing.  
When moving demands and supplies, we update $\set{\alpha_v: v \in \calV}$ and $\set{\beta_v: v \in \calV}$ accordingly.  Keep in mind that we always maintain the property that $\alpha_v \leq \beta_v$ for every $v \in {\calV}$; we do not change $\alpha_{\calV}$ and $\beta_{\calV}$ (we do not change the total demands or supplies in $\calV$).  After the moving process for $T$, we add $\ceil{\alpha_v}$ open facilities at $v$ for every $v \in \calV$.  We shall compare $\sum_{v \in \calV}\ceil{\alpha_v}$ to $\alpha_\calV$.

To define the moving process for $T = ({\calV}, E, r)$, we give each edge in $E$ a rank as follows. An edge $e = (v, v') \in E$ has length $L_e := d(v, v')$. Sort edges in $E$ according to their lengths; assume $e_1, e_2, \cdots, e_{|{\calV}| - 1}$ is the ordering. Let the rank of $e_1$ be 1. For each $t = 2, 3, \cdots, e_{|{\calV}|-1}$, if $L_{e_t} \leq 2\sum_{s = 1}^{t-1}L_{e_s}$, then let the rank of $e_t$ be the rank of $e_{t-1}$; otherwise let the rank of $e_t$ be the rank of $e_{t-1}$ plus 1. Let $h$ be the rank of $e_{|{\calV}|-1}$. For each $i \in [h]$, let $E_i$ be the set of rank-$i$ edges in $E$; for $i = 0, 1, \cdots, h$, let $E_{\leq i} = \union_{i' \leq i}E_{i'}$ be the set of edges of rank at most $i$.

\begin{claim}
\label{claim:lengths-of-same-rank}
For any $i \in [h]$ and $e, e' \in E_i$, we have $L_e/L_{e'} \leq 3^{|{\calV}|-1}$.
\end{claim}

\begin{proof}
It suffices to prove the lemma for the case where $e'$ is the shortest rank-$i$ edge, and $e$ is the longest rank-$i$ edge.  Suppose $e' = e_{t'}$ and $e = e_t$ for $t' < t$.  Let $L = \sum_{e'' \in E_{\leq i-1}}L_{e''}$. Then, $L_{e'} > 2L$. For every $s \in \set{t', t'+1, \cdots, t-1}$, we have $L_{e_{s+1}} \leq 2(L + L_{e_{t'}} + L_{e_{t'+1}} + \cdots +L_{e_s})$.  Thus, $L + L_{e_{t'}} + L_{e_{t'+1}} + \cdots +L_{e_s} + L_{e_{s+1}} \leq 3(L + L_{e_{t'}} + L_{e_{t'+1}} + \cdots +L_{e_s})$. Thus, $L_e \leq 3^{t - t'}(L + L_{e'}) < \frac{3}{2}\cdot 3^{t-t'}L_{e'} \leq 3^{|\calV|-1}L_{e'}$ as $t - t' \leq |\calV|-2$.
\end{proof}

For every $i \in \set{0, 1, \cdots, h}$, we call the set of vertices in a connected component of $(\calV, E_{\leq i})$ a level-$i$ set. The family of level-$i$ sets forms a partition of ${\calV}$; and the union of families over all $i \in \set{0, 1, \cdots, h}$ is a laminar family.  For every $i \in [h]$ and every level-$i$ set $\calA$, we check if Constraint~\eqref{LPC:compact} is satisfied for $\calB = \calU_\calA$ and every $\calJ \subseteq \calC$ (recall that this can be checked efficiently). If not, we find a violation of Constraint~\eqref{LPC:compact}; from now on, we assume Constraint~\eqref{LPC:compact} holds for all these rectangles $(\calB, \calJ)$.

\begin{claim}
\label{claim:identify-isl}
If a level-$i$ set ${\calA}$ does not contain the root $r$, then $d(\calA, \calC^* \setminus \calA) \geq \frac{L'}{2}$, where $L'$ is the length of the shortest edge in $E_{i+1}$. 
\end{claim}

\begin{proof}
See Figure~\ref{figure:tree} for the notations used in the proof. Let $v$ be the highest vertex in $\calA$ according to $T$, and $L = \sum_{e \in E_{\leq i}}L_e$ be the total length of edges of rank at most $i$. 

Notice that $\calC^* \setminus \calA = (\calC^* \setminus \Lambda_T(v)) \cup (\Lambda_T(v) \setminus \calA)$. $d(v, \calC^* \setminus \Lambda_T(v)) = d(v, \rho_T(v)) \geq L'$ since $T$ is a neighborhood tree and the rank of $(v, \rho_T(v))$ is at least $i + 1$.  Thus $d(\calA, \calC^* \setminus \Lambda_T(v)) \geq L' - L \geq \frac{L'}{2}$ as the distance from $v$ to any vertex in $\calA$ is at most $L$. We now bound $d(\calA, \Lambda_T(v) \setminus \calA)$.  Consider each connected component in $(\Lambda_T(v) \setminus \calA, E_{\leq i})$. Let $\calA'$ be the set of vertices in the component and $v'$ be its root. Since $d(v', \rho_T(v')) \geq L'$, we have $d(v', v) \geq L'$ as $\rho_T(v')$ is the nearest representative to $v'$ in $\calC^* \setminus \Lambda_T(v') \ni v$. Since each of $\calA$ and $\calA'$ is connected by edges in $E_{\leq i}$, $v \in \calA, v' \in \calA'$ and the total length of edges in $E_{\leq i}$ is $L$, we have that $d(\calA, \calA') \geq L' - L \geq \frac{L'}{2}$.  As this is true for any such $\calA'$, we have $d(\calA, \Lambda_T(v) \setminus \calA) \geq \frac{L'}{2}$, which, combined with $d(\calA, \calC^* \setminus \Lambda_T(v)) \geq \frac{L'}{2}$, implies the lemma.
\end{proof}

\begin{figure}
\centering
\includegraphics[width=0.5\textwidth]{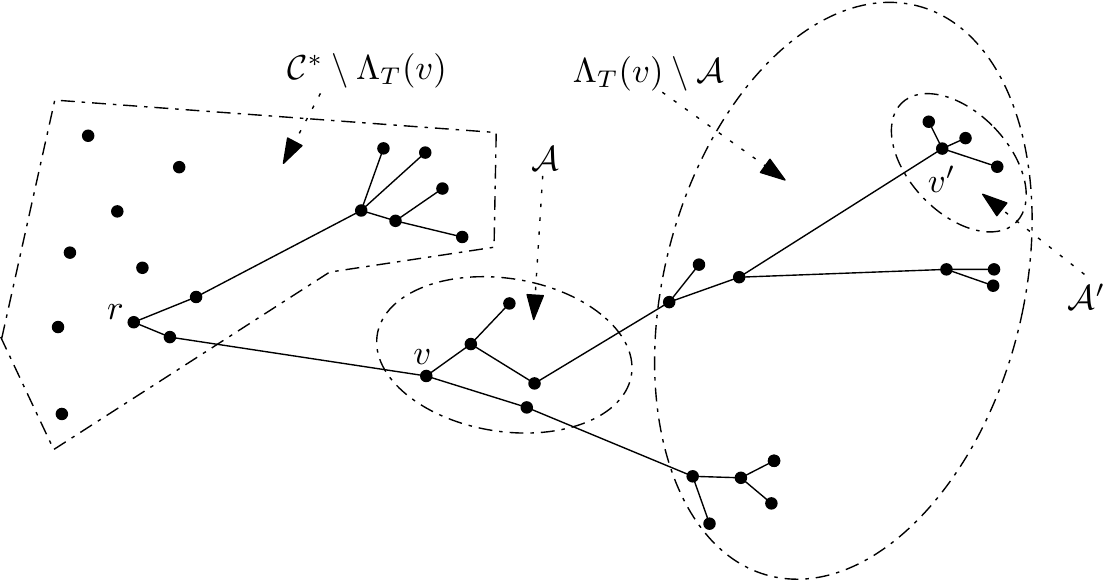}
\caption{Notations used in the proof of Claim~\ref{claim:identify-isl}. Solid circles are client representatives in $\calC^*$ and solid lines give a neighborhood tree $T$.}
\label{figure:tree}
\end{figure}

Recall that the family of all level-$i$ sets, over all $i=0, 1, 2,\cdots, h$ form a laminar family. Level-$0$ sets are singletons and the level-$h$ set is the whole set ${\calV}$.     Our moving operation is level-by-level: for every $i = 1, 2, \cdots, h$ in this order, for every level-$i$ set $\calA \subseteq \calV$, we define a moving process for $\calA$, in which we move demands and supplies within vertices in $\calA$.    After the moving operation for $\calA$, we guarantee the following properties.   

If $r \not\in \calA$, then either
\begin{enumerate}[label=(N\arabic*),labelindent=*, leftmargin=*]
\item all but one vertices $v \in {\calA}$ have $\alpha_v = \beta_v \in \Z_*$;  or
\label{property:all-but-one-integral}
\item every vertex $v \in {\calA}$ has $\beta_v \geq \ceil{\alpha_v} - 1/\ell$.
\label{property:all-good}
\end{enumerate}
If $r \in \calA$, then 
\begin{enumerate}[label=(I\arabic*),labelindent=*, leftmargin=*]
\item every vertex $v \in {\calA} \setminus \set{r}$ has $\beta_v \geq \ceil{\alpha_v}-1/\ell$. \label{property:set-containing-root}
\end{enumerate}

The above properties hold for all level-$0$ sets: they are all singletons;  Property~\ref{property:all-but-one-integral} holds if $r \notin \calA$ and Property~\ref{property:set-containing-root} holds if $r \in \calA$. Now,  suppose the properties hold for all level-$(i-1)$ sets. We define a moving operation for a level-$i$ set ${\calA}$ after which ${\calA}$ satisfies the properties. 

The first step is a collection step, in which we collect demands and supplies from ${\calA}$. For every $v \in \calA \setminus \set{r}$ such that $\beta_v < \ceil{\alpha_v} - 1/\ell$, we collect $\lfrac{\alpha_v}$ units of demand and $\beta_v - \floor{\alpha_v}$ units of supply from $v$ and keep them in a temporary holder. For all vertices $v \in \calA$ with $\beta_v > \ceil{\alpha_v}$, we collect $\beta_v - \ceil{\alpha_v}$ units of supply from $v$. Now, we have $\ceil{\alpha_v}-1/\ell \leq \beta_v \leq \ceil{\alpha_v}$ for every $v \in \calA \setminus \set{r}$. 

The second step is a redistribution step,  in which we move the demand and supply in the temporary holder back to ${\calA}$. If $r \in \calA$, we simply move the demand and the supply in the holder to $r$ and terminate the process. $\calA$ will satisfy Property~\ref{property:set-containing-root}. From now on we assume $r \notin \calA$. We try to move the demand and the supply in the holder to each $v \in \calA$ continuously until we have $\alpha_v = \beta_v \in \Z_*$: we first move demand from the holder to $v$ until $\alpha_v = \beta_v$, then move demand and supply at the same rate until $\alpha_v = \beta_v \in \Z_*$.  If we succeeded in making all vertices $v \in {\calA}$ satisfy $\alpha_v = \beta_v \in \Z_*$, then we can move the remaining supplies and demands in the holder to an arbitrary vertex in ${\calA}$. In this case ${\calA}$ satisfies Property~\ref{property:all-but-one-integral}. Suppose we failed to make $\alpha_v = \beta_v \in \Z_*$ for some $v \in {\calA}$.  The failure is due to the insufficient demand in the holder: we have collected at least the same amount of supply as demand; in the redistribution step, we either move the demand from the holder or move the demand and the supply at the same rate. We then move all the remaining supply in the holder to an arbitrary vertex $v \in {\calA}$. Notice that during the continuous redistribution process for $v$, we always maintain the property that $\ceil{\alpha_v}-1/\ell \leq \beta_v \leq \ceil{\alpha_v}$.  Moving the remaining supply to an arbitrary vertex $v$ also maintain the property that $\ceil{\alpha_v} - 1/\ell \leq \beta_v$. Thus ${\calA}$ will satisfy Property~\ref{property:all-good} in the end.

After we finished the moving operation for the level-$h$ set $\calV$,  our set ${\calV}$ satisfies Property~\ref{property:set-containing-root} as $r \in \calV$. Thus $\sum_{v \in {\calV}}\ceil{\alpha_v} \leq \sum_{v \in {\calV} \setminus \set{r}} (\beta_v + 1/\ell) + \alpha_r + 1 \leq \beta_\calV  + 1 + (\cardinal{\calV}-1)/\ell \geq \frac{2\ell-1}{(\ell-1)^2}\beta_{\calV}$ as $\beta_\calV \geq (\cardinal{\calV} - 1)(1-1/\ell)$ and $\cardinal{\calV} \geq \ell$. Taking this sum over all trees in $\bbT$, we have that the number of open facilities is at most $\frac{2\ell-1}{(\ell-1)^2}k$. By setting $\ell = \ceil{3/\eps}$, the number of open facilities  is at most $(1+\eps)k$.

It suffices to bound the moving cost for $T$.
\begin{lemma}
\label{lemma:moving-cost-for-T}
The moving cost of the operation for $T = ({\calV}, E, r) \in \bbT$ is at most $$\exp\left(O(\ell^2)\right)\left(D(\calU_{\calV\setminus \set{r}}) + D'(\calU_{\calV \setminus \set{r}})\right).$$
\end{lemma}

\begin{proof}
Consider the moving process for a level-$i$ set ${\calA}$.  Suppose we collected some demand from $v \in \calA$. It must be the case that $v \neq r$ and $\beta_v < \ceil{\alpha_v} - 1/\ell$ before the collection, as otherwise we would not collect demand from $v$. If we let $\calA' \subseteq \calA$ be the level-$(i-1)$ set containing $v$, then $\calA'$ must satisfy $r \notin \calA'$ and Property~\ref{property:all-good} by the induction assumption. This implies that we did not collect demands from any other vertices  in $\calA'$.  Notice that  $\beta_v < \ceil{\alpha_v}-1/\ell$ implies $\alpha_v > \floor{\beta_v}$ and $\rfrac{\beta_v} > 1/\ell$.  Then, $\alpha_{\calA'} > \floor{\beta_{\calA'}}$ and $\rfrac{\beta_{\calA'}} > 1/\ell$ as all vertices $v' \in {\calA'} \setminus \set{v}$ have $\alpha_{v'} = \beta_{v'} \in \Z_*$. Since we never moved demands or supplies in or out of ${\calA'}$ before, we have $\alpha_{\calA'} = y'_\calS$ and $\beta_{\calA'} = y_\calS$, where $\calS = \calU_{\calA'}$. Then $y'_\calS >\floor{y_\calS}$ and $\rfrac{y_\calS} > 1/\ell$. 

  As we assumed that Constraint~\eqref{LPC:compact} is satisfied for $\calB = \calS$ and every $\calJ \subseteq \calC$, we can apply Lemma~\ref{lemma:bounding-moving-cost} to show that $\lfrac{y'_\calS}\rfrac{y_\calS} d({\calA'}, \calC^* \setminus {\calA'}) \leq \frac{4}{u}D_\calS + \frac{4\ell+2}{u} D'_\calS$. The demands collected from $v$ will be moved to vertices in ${\calA}$.  The moving distance is at most $\sum_{e \in E_{\leq i}}L_e \leq \cardinal{{\calV}} 3^{\cardinal{\calV}}L'$ by Claim~\ref{claim:lengths-of-same-rank}, where $L'$ is the length of the shortest edge in $E_i$. Now, by Claim~\ref{claim:identify-isl}, $L' \leq 2d(\calA', \calC^* \setminus \calA')$. Thus, the moving distance is at most 
\vspace*{-0.6\abovedisplayskip}
\[2\cardinal\calV3^{\cardinal\calV}d(\calA', \calC^* \setminus \calA') \leq \frac{2}{4u}\frac{\cardinal{{\calV}}3^{\cardinal{{\calV}}}}{\lfrac{y'_\calS}\rfrac{y_\calS}} (2D_\calS + (2\ell+1) D'_\calS).
\]
\vspace*{-0.6\belowdisplayskip}

\noindent Notice that $\rfrac{y_\calS} > 1/\ell$ and $\cardinal{\calV} \leq \ell^2$. The distance is at most $\frac{\exp\left(O(\ell^2)\right)}{u\lfrac{y'_\calS}} (D_\calS + D'_\calS)$. As we moved $\lfrac{y'_\calS}$ units of demand from $\calA'$, the moving cost is at most $\exp\left(O(\ell^2)\right)(D_\calS + D'_\calS)$.  

Taking the sum of the upper bounds over all level-$(i-1)$ sets $\calA' \subseteq \calA \setminus \set{r}$,  the cost is at most $\exp\left(O(\ell^2)\right)\left(D(\calU_{\calA\setminus \set{r}}) + D'(\calU_{\calA \setminus \set{r}})\right)$.  Taking the sum over $i \in [h]$ and all level-$i$ sets $\calA$, the cost is at most $\exp\left(O(\ell^2)\right)\left(D(\calU_{\calV\setminus \set{r}}) + D'(\calU_{\calV \setminus \set{r}})\right)$, as the number $h$ of levels is absorbed by $\exp\left(O(\ell^2)\right)$. This finishes the proof of Lemma~\ref{lemma:moving-cost-for-T}.
\end{proof}

Finally, taking the bound over all neighborhood trees $T = ({\calV}, E, r)$, the moving cost is at most $\exp(O(\ell^2))(D_\calF + D'_\calF)$ due to Property~\ref{property:trees-partitioning} and the fact that $\set{\calU_v : v \in \calC^*}$ forms a partition of $\calF$.  Since $D_\calF = D'_\calF = \LP$, the moving cost is at most $\exp(O(\ell^2))\LP$. 

This finishes the proof of Theorem~\ref{theorem:main} for the case $\cardinal{\calC^*} \geq \ell$. When $\cardinal{\calC^*} < \ell$, we only build one neighborhood tree $(\calC^*, E, r)$. Any minimum spanning tree over $\calC^*$ will be a neighborhood tree.  We run the algorithm for this neighborhood tree. The argument for moving cost still works; it suffices to bound the number of open facilities. After the moving process, we have $\beta_v \geq \ceil{\alpha_v} - 1/\ell$ for every $v \in \calC^*\setminus \set{r}$. Also $\beta_{\calC^*}\leq k$. Thus, $\beta_r \leq k - \beta_{\calC^* \setminus \set{r}} \leq k-\sum_{v \in \calC^*\setminus \set{r}}(\ceil{\alpha_v}-1/\ell) \leq k-\sum_{v \in \calC^*\setminus \set{r}}\ceil{\alpha_v} + (\ell-2)/\ell$ as $\cardinal{\calC^*} < \ell$. Thus, $\ceil{\alpha_r} \leq \ceil{\beta_r} \leq k - \sum_{v \in \calC^* \setminus \set{r}}\ceil{\alpha_v} + 1$, implying $\sum_{v \in \calC^*}\ceil{\alpha_v} \leq k+1$. Thus, the number of open facilities is at most $k+1 \leq \ceil{(1+\eps)k}$.  To finish the proof of Theorem~\ref{theorem:main}, it remains to prove the two technical lemmas.

\subsection{Proofs of Technical Lemmas}
\label{subsec:omitted-proofs}
\ifdefined\CR
\ \\
\begin{proofof}{\bf Proof of Lemma~\ref{lemma:xb-times-xbb-large}}
\else
\begin{proof}[\bf Proof of Lemma~\ref{lemma:xb-times-xbb-large}]
\fi
For simplicity we let $y = y_\calB, y'=y'_\calB$ and $x_j = x_{\calB, j}$ for every $j \in \calC$.   Throughout the proof,  $y$ and $y'$ are fixed.

We assume $\calC = [n]$ and $1 \geq x_1 \geq x_2 \geq \cdots \geq x_n \geq 0$.    Let  $\bar{f}(p) = \min\set{f(p, y), uy'}$ for every integer $p \in [0, n]$. Notice that $\bar f$ is a non-decreasing concave function as $f(\cdot, y)$ is concave and $uy'$ is independent of $p$.   The conjunction of Constraint~\eqref{LPC:compact} and $y' = \sum_{j = 1}^n x_j/u$ is equivalent to $\sum_{j = 1}^{p} x_j \leq \bar f(p)$ for every $p \in [n]$. 

Let $g:[0,1] \to \R$ be any second-order differentiable concave function such that $g(0) = 0$.   We shall show that $\sum_{j = 1}^{n} g(x_j) \geq \sum_{j=1}^{n} g(x^*_j)$, where $x^*_j = \bar f(j) - \bar f(j-1)$ for every $j \in [n]$.  


We use $g'$ and $g''$ to denote the first-order and second-order derivative functions of $g$ respectively.  For any $x \in [0, 1]$, let $\psi(x) = \cardinal{\set{j \in \calC:x_j \geq x}}$. Then 
\ifdefined\CR
\begin{align*}
\sum_{j \in \calC}g(x_j) &= \int_0^1\psi(x)g'(x)dx \\
&=\int_0^1\left(g'(0)+\int_0^xg''(t)dt\right)\psi(x)dx\\
&=g'(0)\int_0^1\psi(x)dx + \int_0^1\left[\int_t^1\psi(x)dx\right]g''(t)dt.
\end{align*}
\else
\begin{align*}
\sum_{j \in \calC}g(x_j) &= \int_0^1\psi(x)g'(x)dx = \int_0^1\left(g'(0)+\int_0^xg''(t)dt\right)\psi(x)dx\\
&=g'(0)\int_0^1\psi(x)dx + \int_0^1\left[\int_t^1\psi(x)dx\right]g''(t)dt.
\end{align*}
\fi

Notice that the first term is equal to $g'(0)\sum_{j \in \calC}x_j = g'(0)uy'$, which is independent of $\vec x := (x_1, x_2, \cdots, x_n)$.  Since $g$ is concave, we have $g''(t) \leq 0$ for every $t \in [0, 1]$.  We show that $Q(t):=\int_t^1\psi(x)dx$ is maximized when $\vec x = \vec x^*:=(x^*_1, x^*_2, \cdots, x^*_n)$, for every $t \in [0, 1]$. 

We now fix $t \in [0, 1]$. Notice that $Q(t) = \sum_{j = 1}^{p_t} (x_j - t)$ where $p_t$ is the largest integer $p$ such that $x_p \geq t$.  Then $Q(t) \leq \bar f(p_t) - tp_t \leq \max_{p=0}^{n} \left(\bar f(p) - tp\right)$.  




We show that $Q(t) = \max_{p=0}^{n} \left(\bar f(p) - tp\right) $ when $\vec x = \vec x^*$. Consider the sequence $x_1 - t, x_2 - t, \cdots, x_n - t$. The sequence is non-increasing; $\bar f(p) - tp$ is the sum of the first $p$ number in the sequence by the definition of $\set{x^*_j}$. Thus,  the sum is maximized when $\bar f(p) - tp$ is the largest number such that $x_p \geq t$. This $p$ is exactly the definition of $p_t$.  Thus, $Q(t)$ is maximized when $\vec x = \vec x^*$.  This proves that $\sum_{i=1}^n g(x_i) \geq \sum_{i=1}^n g(x^*_i)$.



Now we let $g(x) \equiv x(1 - x)$. Then $g(0) = g(1) = 0$. Thus, $\sum_{j \in \calC}g(x^*_j) = \floor{\frac{u\lfrac{y'}}{\lfrac{y}}}g\left(\lfrac{y}\right) + g\left(\lfrac{\frac{u\lfrac{y'}}{\lfrac{y}}}\lfrac{y}\right)$. By the concavity of $g$ and $g(0) = 0$, we have $g\left(\lfrac{\frac{u\lfrac{y'}}{\lfrac{y}}}\lfrac{y}\right) \geq \lfrac{\frac{u\lfrac{y'}}{\lfrac{y}}}g\left(\lfrac{y}\right)$. Thus $\sum_{j \in \calC}g(x_j)\geq \left(\floor{\frac{u\lfrac{y'}}{\lfrac{y}}}+\lfrac{\frac{u\lfrac{y'}}{\lfrac{y}}}\right) g\left(\lfrac{y}\right) = \frac{u\lfrac{y'}}{\lfrac{y}}g(\lfrac{y}) = \frac{u\lfrac{y'}}{\lfrac{y}}\lfrac{y}(1-\lfrac{y}) = u\lfrac{y'}\rfrac{y}$. The last equation used the fact that $\lfrac{y} + \rfrac{y} = 1$ if $y$ is fractional and $\lfrac{y'} = 0$ if $y$ is integral. 
\ifdefined\CR
\end{proofof}
\else
\end{proof}
\fi

\ifdefined\CR
\smallskip
\begin{proofof}{\bf Proof of Lemma~\ref{lemma:constructing-neighborhood-trees}}
\else
\begin{proof}[\bf Proof of Lemma~\ref{lemma:constructing-neighborhood-trees}]
\fi
The first step is a simple iterative process. We maintain a spanning forest of rooted trees for $\calC^*$. Initially, we have $\cardinal{\calC^*}$ singletons. At each iteration, we arbitrarily choose a tree $T=({\calV}, E, r)$ of size less than $\ell$. Let $v^* = \arg\min_{v \in \calC^* \setminus {\calV}}d(r, v)$ be the nearest neighbor of $r$ in $\calC^* \setminus {\calV}$. Assume $v^*$ is in some rooted tree $T' = ({\calV}', E', r')$. Then, we merge $T$ and $T'$ by adding an edge $(r, v^*)$, and let $v^*$ be the parent of $r$. i.e, the new tree will be $({\calV} \cup {\calV}', E \cup E' \cup\set{(r, v^*)}, r')$. The process ends when all rooted trees have size at least $\ell$. 

Now we show that every rooted tree in the spanning forest is a neighborhood tree. Initially all trees are trivially neighborhood trees; it suffices to prove that the new tree formed by merging two neighborhood trees is also a neighborhood tree. Consider two neighborhood trees $T = ({\calV}, E, r)$ and $T' = ({\calV}', E', r')$, and suppose we obtain a merged tree $T''$ by adding an edge $(r, v^*)$ for some $v^* \in {\calV}'$.  Then, for every $v \in {\calV} \setminus \set{r}$ we have $d(v, \calC^* \setminus \Lambda_{T''}(v)) = d(v, \rho_{T''}(v))$ since $\Lambda_{T''}(v) = \Lambda_T(v), \rho_{T''}(v) = \rho_{T}(v)$ and $T$ is a neighborhood tree. Also, we have $d(r, \calC^* \setminus \Lambda_{T''}(r)) = d(r, \rho_{T''}(r))$ since $\Lambda_{T''}(r) = {\calV}$ and $\rho(T'')(r) = v^*$ is the nearest neighbor of $r$ in $\calC^* \setminus {\calV}$. Finally, for every $v \in {\calV}' \setminus \set{r'}$, we have $d(v, \calC^* \setminus T''_{v}) = d(v, \rho_{T''}(v))$ since $\rho_{T''}(v) = \rho_{T'}(v) \in \calC^* \setminus \Lambda_{T''}(v) \subseteq \calC^* \setminus \Lambda_{T'}(v)$ and $T'$ is a neighborhood tree.

All neighborhood trees we constructed have size at least $\ell$.  However, they might have size much larger than $\ell^2$.  Thus, we need to break a large neighborhood tree.  Focus on a neighborhood tree $T = ({\calV}, E, r)$ of size more than $\ell^2$ and consider its growth in the iterative process. Initially it contains only a single vertex $r$. During the execution of the process, we merge it with some neighborhood tree $T'$ of size less than $l$, by ``hanging'' $T'$ at some vertex of $T$.  We call $T'$ a treelet. When we hang $T' = ({\calV}', E', r')$ at some vertex $v^*$ of $T$, we have $d(r', v^*) = d(r', \calC^* \setminus {\calV}')$.

Let $\tilde T$ be the tree obtained from $T$ by contracting vertices of the same treelet into a super-node (for convenience, the root $r$ is a treelet). Thus, $\tilde T$ is a tree rooted at $r$, where each super-node corresponds to a treelet. Let the weight of a super-node be the size of its correspondent treelet.  Consider the deepest super-node of $\tilde T$ such that the total weight of the sub-tree rooted at this super-node is at least $\ell(\ell-1)$. The treelet $T' = ({\calV}', E', r')$ correspondent to this super-node has size at most $\ell-1$. Then there must be a vertex $v \in {\calV}'$ such that the total weight hanging at this vertex is at least $\frac{\ell(\ell-1)-(\ell-1)}{\ell-1} = \ell - 1$.  Then, we break $T$ into two parts: the sub-tree $T^1$ of $T$ rooted at $v$, and the sub-tree $T^2$ of $T$ obtained by removing descendants of $v$. The two sub-trees share the vertex $v$.  Then $T^2$ obviously a neighborhood tree since it is the sub-tree of $T$ rooted at $v$ and $T$ is a neighborhood tree.  Also, $T^1$ is neighborhood tree since we can obtain $T^1$ from $r$ by repeatedly hanging treelets.   $T^2$ has size at least $\ell$ and at most $\ell(\ell-1)$ and $T^1$ has size at least $\ell$.  Now we let $T \leftarrow T^1$ and repeat the process until the size of $T$ is at most $\ell^2$. 

All neighborhood trees have size between $\ell$ and $\ell^2$ and the union of all neighborhood trees cover all vertices of $\calC^*$.  Every vertex of $\calC^*$ can appear at most once as a non-root of some neighborhood tree. 
\ifdefined\CR
\end{proofof}
\else
\end{proof}
\fi

\section{Integrality Gap for the Rectangle LP}
\label{sec:lower-bound}

In this section, we show that the integrality gap of the rectangle LP is $\Omega(\log n)$ if we are only allowed to open $k$ facilities.  

The instance is as follows. Let $G = (V, E)$ be a degree-3 expander of size $|V| = u$ and the metric is the graph metric defined by $G$.  Each vertex $i \in V$ is a facility location and there are $u + 1$ clients at $i$. We are allowed to open $k = u  + 1$ facilities and each open facility has capacity $u$.  Thus, there are $|\calC| = n := u(u+1)$ clients. Previously we assumed $\calF = \calC$. Since we only need to keep one facility location for every set of $u+1$ co-located clients, we can assume $\calF = V$.

We first show that the cost of the optimum integral solution is $\Omega(u\log u)$. Intuitively,  an optimum solution opens two facilities at some $i \in \calF$ and one facility at each of the other locations.  In this case, the cost of the optimum integral solution is $\Omega(u \log u)$.  However, it is a little bit involved to prove this intuition and thus we shall avoid it. 

If there are more than $\log u$ locations in $\calF$ without open facilities, then the cost of the facilities is at least $(u+1)\log u= \Omega(u \log u)$.  Thus we assume there are less than $\log u$ locations without open facilities. If a location contains one open facility,  we can assume it is connected by  $u$ clients at this location. Then we remove this open facility and the $u$ clients.  Now, we have at least one client left at each location and at most $\log u + 1$ open facilities left. It is easy to see that these clients will cost $\Omega(u \log u)$.

We now turn to prove that the optimum fractional solution to the rectangle LP is $O(u)$. We open $y_i = 1 + 1/u$ facilities at each location $i \in \calF$. Fix a client $j \in \calC$ co-located with $i$. Let $x_{i,j} = 1-3\gamma/u$ for some large enough constant $\gamma$. For each location $i'$ that is a neighbor of $i$, let $x_{i',j} = \gamma/u$. For all other locations $i'$ we have $x_{i',j} = 0$. This is obviously a valid solution to the basic LP.  Every client $j$ has cost $3\gamma/u$. Thus the total cost is $3\gamma(u+1) = O(u)$.

We now show that Constraint~\eqref{LPC:compact} holds. Focus on a non-empty set $\calB \subsetneq \calF$ of locations.  Let $t = |\calB|$ and $q = y_\calB = t(1+1/u)$.  We identify $\calC$ with $[n]$ and assume $x_{\calB,1} \geq x_{\calB,2} \geq \cdots \geq x_{\calB, n}$. It suffices to prove that for every $p \in [n]$, we have $x_{\calB, [p]} \leq f(p, q)$.  

Given a concave function $g$ on $\set{0, 1, 2, \cdots, n}$,  we say an integer $t \in [0, n]$ is a corner point if either $t \in \set{0, n}$ or $2g(t) > g(t-1) + g(t+1)$. Notice that $x_{\calB, [p]}$ is a concave function of $p$.  Since the $u + 1$ clients at each location $i$ are symmetric, a corner point of $x_{\calB, [\cdot]}$ must be a multiply of $u+1$. $f(\cdot, q)$ has four corner points: $0$, $ut, ut + u$ and $n$. 

Assume $x_{\calB,[p]} > f(p, q)$ for some $p$. Then this must hold for some $p'$ which is either a corner point of $x_{\calB,[\cdot]}$, or a corner point of $f(\cdot, q)$. As $(x,y)$ is a valid solution to the basic LP, we have $x_{\calB, [p]} \leq \min\set{p, qu}$ for every $p$.  From the definition of $f$, it must be the case that $p' \in (u\floor{q}, u \ceil{q}) = (ut, ut + u)$. The only possibility for $p'$ is $p' = (u+1)t$. In this case, $[p']$ contain the $(u+1)t$ clients co-located with facilities in $\calB$.


%


First assume $t \leq u/2$.   Then, $x_{\calB, [p']} \leq (u+1)(t - \alpha t\frac{\gamma}{u}) \leq ut + t - \gamma\alpha t$, where $\alpha$ is the expansion constant of $G$. The inequality is due to the fact that $E(\calB, \calF \setminus \calB) \geq \alpha |\calB| = \alpha t$. If $\gamma\geq 1/\alpha$, we have $x_{\calB, [p']} \leq ut = f(ut, t) \leq f(p', q)$.   Now assume $t > u/2$. Then $x_{\calB, [p']}\leq (u+1)(t - \alpha(u-t)\frac{\gamma}{u})$. If $\gamma \geq 1/\alpha$, we have $x_{\calB, [p']} \leq (u+1)t-(u - t) = ut +2t - u \leq ut + t^2/u = f(p', q)$.  This leads to a contradiction.  Thus, the fractional solution satisfies Constraint~\eqref{LPC:compact}.

Overall, we have showed an $\Omega(\log u) = \Omega(\log n)$ integrality gap. 
\section{Discussion}
\label{section:discussion}

In this paper, we introduced a novel rectangle LP relaxation for uniform \CKM and gave a rounding algorithm which produces $\exp\left(O(1/\eps^2)\right)$-approximate solutions by opening $(1+\eps)k$ facilities.  This is beyond the approximability of the natural LP relaxation, as it has unbounded integrality gap even if $(2-\eps)k$ facilities are allowed to be opened.  There are many related open problems. 

First, can our rectangle LP give a constant approximation for uniform \CKM by violating capacity constraints by $1+\eps$?  The difficulty of this problem seems to be that each facility has a capacity constraint and we need to guarantee that none of them is violated by too much. While in our problem, we are only concerned with one cardinality constraint.

Second, can we extend our result to non-uniform \CKM?  Without uniform capacities, we can not even prove Theorem~\ref{theorem:soft-hard-same}. Also, Constraint~\eqref{LPC:compact} crucially used the uniformity.  Without it, we may need to generalize Constraint~\eqref{LPC:compact}.

Finally, can we obtain a true constant approximation for uniform \CKM?  
\ifdefined \CR
After submitting the extended abstract of the paper, we found an example showing that the integrality gap of the rectangle LP is $\Omega(\log n)$, if the cardinality constraint can not be violated. We defer the proof of the integrality gap to the full version of the paper. Thus, a stronger LP is needed to obtain a true constant approximation.
\else
As we showed that the integrality gap of the rectangle LP is $\Omega(\log n)$, if the cardinality constraint can not be violated, we need a stronger LP to obtain a true constant approximation. 
\fi

\section*{Acknowledgement} I want to thank Ola Svensson for pointing out a simplification for the rectangle LP.  I also thank anonymous reviewers for many useful comments over the extended abstract version of the paper. 

\bibliographystyle{plain}
\bibliography{reflist}

\begin{thebibliography}{10}

\bibitem{AAB10}
Ankit Aggarwal, L.~Anand, Manisha Bansal, Naveen Garg, Neelima Gupta, Shubham
  Gupta, and Surabhi Jain.
\newblock A 3-approximation for facility location with uniform capacities.
\newblock In {\em Proceedings of the 14th International Conference on Integer
  Programming and Combinatorial Optimization}, IPCO'10, pages 149--162, Berlin,
  Heidelberg, 2010. Springer-Verlag.

\bibitem{ASS14}
Hyung-Chan An, Mohit Singh, and Ola Svensson.
\newblock {LP}-based algorithms for capacitated facility location.
\newblock In {\em Proceedings of the 55th Annual IEEE Symposium on Foundations
  of Computer Science, FOCS 2014}.

\bibitem{AGK01}
V.~Arya, N.~Garg, R.~Khandekar, A.~Meyerson, K.~Munagala, and V.~Pandit.
\newblock Local search heuristic for k-median and facility location problems.
\newblock In {\em Proceedings of the thirty-third annual ACM symposium on
  Theory of computing}, STOC '01, pages 21--29, New York, NY, USA, 2001. ACM.

\bibitem{BGG12}
Manisha Bansal, Naveen Garg, and Neelima Gupta.
\newblock A 5-approximation for capacitated facility location.
\newblock In {\em Proceedings of the 20th Annual European Conference on
  Algorithms}, ESA'12, pages 133--144, Berlin, Heidelberg, 2012.
  Springer-Verlag.

\bibitem{Byr07}
J.~Byrka.
\newblock An optimal bifactor approximation algorithm for the metric
  uncapacitated facility location problem.
\newblock In {\em APPROX '07/RANDOM '07: Proceedings of the 10th International
  Workshop on Approximation and the 11th International Workshop on
  Randomization, and Combinatorial Optimization. Algorithms and Techniques},
  pages 29--43, Berlin, Heidelberg, 2007. Springer-Verlag.

\bibitem{BFR13}
Jaroslaw Byrka, Krzysztof Fleszar, Bartosz Rybicki, and Joachim Spoerhase.
\newblock Bi-factor approximation algorithms for hard capacitated k-median
  problems.
\newblock In {\em Proceedings of the 26th Annual ACM-SIAM Symposium on Discrete
  Algorithms (SODA 2015)}.

\bibitem{BPR15}
Jaroslaw Byrka, Thomas Pensyl, Bartosz Rybicki, Aravind Srinivasan, and Khoa
  Trinh.
\newblock An improved approximation for $k$-median, and positive correlation in
  budgeted optimization.
\newblock In {\em Proceedings of the 26th Annual ACM-SIAM Symposium on Discrete
  Algorithms (SODA 2015)}.

\bibitem{CG99}
M.~Charikar and S.~Guha.
\newblock Improved combinatorial algorithms for the facility location and
  k-median problems.
\newblock In {\em In Proceedings of the 40th Annual IEEE Symposium on
  Foundations of Computer Science}, pages 378--388, 1999.

\bibitem{CGT99}
M.~Charikar, S.~Guha, E.~Tardos, and D.~B. Shmoys.
\newblock A constant-factor approximation algorithm for the k-median problem
  (extended abstract).
\newblock In {\em Proceedings of the thirty-first annual ACM symposium on
  Theory of computing}, STOC '99, pages 1--10, New York, NY, USA, 1999. ACM.

\bibitem{CL12}
Moses Charikar and Shi Li.
\newblock A dependent lp-rounding approach for the k-median problem.
\newblock In {\em Proceedings of the 39th International Colloquium Conference
  on Automata, Languages, and Programming - Volume Part I}, ICALP'12, pages
  194--205, Berlin, Heidelberg, 2012. Springer-Verlag.

\bibitem{CS04}
F.~A. Chudak and D.~B. Shmoys.
\newblock Improved approximation algorithms for the uncapacitated facility
  location problem.
\newblock {\em SIAM J. Comput.}, 33(1):1--25, 2004.

\bibitem{CW05}
Fabian~A. Chudak and David~P. Williamson.
\newblock Improved approximation algorithms for capacitated facility location
  problems.
\newblock {\em Math. Program.}, 102(2):207--222, March 2005.

\bibitem{CR05}
Julia Chuzhoy and Yuval Rabani.
\newblock Approximating k-median with non-uniform capacities.
\newblock In {\em In SODA ’05}, pages 952--958, 2005.

\bibitem{GL13}
Dion Gijswijt and Shanfei Li.
\newblock Approximation algorithms for the capacitated k-facility location
  problems.
\newblock {\em CoRR}, abs/1311.4759, 2013.

\bibitem{GK98}
S~Guha and S~Khuller.
\newblock Greedy strikes back: Improved facility location algorithms.
\newblock In {\em Journal of Algorithms}, pages 649--657, 1998.

\bibitem{JMM03}
K.~Jain, M.~Mahdian, E.~Markakis, A.~Saberi, and V.~V. Vazirani.
\newblock Greedy facility location algorithms analyzed using dual fitting with
  factor-revealing {LP}.
\newblock {\em J. ACM}, 50:795--824, November 2003.

\bibitem{JMS02}
K.~Jain, M.~Mahdian, and A.~Saberi.
\newblock A new greedy approach for facility location problems.
\newblock In {\em Proceedings of the thiry-fourth annual ACM symposium on
  Theory of computing}, STOC '02, pages 731--740, New York, NY, USA, 2002. ACM.

\bibitem{JV01}
K~Jain and V.~V. Vazirani.
\newblock Approximation algorithms for metric facility location and k-median
  problems using the primal-dual schema and {L}agrangian relaxation.
\newblock {\em J. ACM}, 48(2):274--296, 2001.

\bibitem{KPR98}
M.~R. Korupolu, C.~G. Plaxton, and R.~Rajaraman.
\newblock Analysis of a local search heuristic for facility location problems.
\newblock In {\em Proceedings of the ninth annual ACM-SIAM symposium on
  Discrete algorithms}, SODA '98, pages 1--10, Philadelphia, PA, USA, 1998.
  Society for Industrial and Applied Mathematics.

\bibitem{Li11}
S.~Li.
\newblock A 1.488 approximation algorithm for the uncapacitated facility
  location problem.
\newblock In {\em Automata, Languages and Programming - 38th International
  Colloquium (ICALP)}, pages 77--88, 2011.

\bibitem{Li14}
Shanfei Li.
\newblock An improved approximation algorithm for the hard uniform capacitated
  k-median problem.
\newblock In {\em APPROX '14/RANDOM '14: Proceedings of the 17th International
  Workshop on Combinatorial Optimization Problems and the 18th International
  Workshop on Randomization and Computation}, APPROX '14/RANDOM '14, 2014.

\bibitem{LS13}
Shi Li and Ola Svensson.
\newblock Approximating k-median via pseudo-approximation.
\newblock In {\em Proceedings of the Forty-fifth Annual ACM Symposium on Theory
  of Computing}, STOC '13, pages 901--910, New York, NY, USA, 2013. ACM.

\bibitem{LV92B}
J.~Lin and J.~S. Vitter.
\newblock Approximation algorithms for geometric median problems.
\newblock {\em Inf. Process. Lett.}, 44:245--249, December 1992.

\bibitem{MYZ06}
M.~Mahdian, Y.~Ye, and J.~Zhang.
\newblock Approximation algorithms for metric facility location problems.
\newblock {\em SIAM J. Comput.}, 36(2):411--432, 2006.

\bibitem{STA97}
D.~B. Shmoys, E.~Tardos, and K.~Aardal.
\newblock Approximation algorithms for facility location problems (extended
  abstract).
\newblock In {\em STOC '97: Proceedings of the twenty-ninth annual ACM
  symposium on Theory of computing}, pages 265--274, New York, NY, USA, 1997.
  ACM.

\bibitem{WS11}
D~Williamson and Shmoys D.
\newblock {\em The Design of Approximation Algorithms}.
\newblock Cambridge University Press, 2011.

\bibitem{ZCY05}
Jiawei Zhang, Bo~Chen, and Yinyu Ye.
\newblock A multiexchange local search algorithm for the capacitated facility
  location problem.
\newblock {\em Math. Oper. Res.}, 30(2):389--403, May 2005.

\end{thebibliography}
\end{document}